\documentclass[11pt]{article}

\usepackage[margin=1in]{geometry}

\usepackage[T1]{fontenc}
\usepackage[utf8]{inputenc}
\usepackage{mathpazo} 
\usepackage{amsfonts}
\usepackage{amsmath, amssymb, amsthm,mathtools}

\usepackage{soul}

\usepackage{nicefrac}

\usepackage[singlespacing]{setspace}

\usepackage{algorithm}
\usepackage[noend]{algpseudocode}
\usepackage{graphicx}
\usepackage{xcolor}
\usepackage{array}
\usepackage[shortlabels]{enumitem}
\usepackage{comment}
\usepackage{float}
\usepackage[backgroundcolor=gray!10,textsize=footnotesize]{todonotes}

\usepackage{standalone} % REQUIRED to link the figure files
\usepackage{tikz}
\usetikzlibrary{decorations.pathreplacing, decorations.markings, arrows.meta, calc, quotes, shapes.geometric}
\usepackage{pgfplots}
\pgfplotsset{compat=newest}
\graphicspath{{Figures/}}

\usepackage{hyperref}
\usepackage{cleveref} 

\usepackage[style=alphabetic,backend=biber,maxbibnames=10]{biblatex}

\newtheorem{theorem}{Theorem}
\newtheorem{corollary}[theorem]{Corollary}
\newtheorem{lemma}[theorem]{Lemma}
\newtheorem*{conjecture}{Conjecture}
\newtheorem{remark}[theorem]{Remark}
\newtheorem{observation}[theorem]{Observation}

\crefname{lemma}{Lemma}{Lemmas}

\newcommand{\tot}{\mathrm{tot}}
\newcommand{\avg}{\mathrm{avg}}

\title{\huge \textbf{Unsplittable Transshipments}\footnote{The first and the last author are supported by the Deutsche Forschungsgemeinschaft (DFG, German Research
Foundation) under Germany´s Excellence Strategy – The Berlin Mathematics
Research Center MATH+ (EXC-2046/1, EXC-2046/2, project ID: 390685689).}}
\author{Srinwanti Debgupta\footnote{Institute of Mathematics, TU Berlin, Germany, \texttt{debgupta@math.tu-berlin.de}} \and Sarah Morell\footnote{Faculty of Mathematics and Computer Science, University of Bremen, Germany, \texttt{morell@uni-bremen.de}} \and Martin Skutella\footnote{Institute of Mathematics, TU Berlin, Germany, \texttt{martin.skutella@tu-berlin.de}}} % Removes the space usually reserved for authors
\date{February 6, 2026}   % Removes the space usually reserved for date

\usepackage[title]{appendix}

\begin{document}

%----------
% Title page
%----------

\clearpage
\maketitle

\begin{abstract}
We introduce the \emph{Unsplittable Transshipment Problem} in directed graphs with multiple sources and sinks. An unsplittable transshipment routes given supplies and demands using at most one path for each source--sink pair. Although they are a natural generalization of single source unsplittable flows, unsplittable transshipments raise interesting new challenges and require novel algorithmic techniques. As our main contribution, we give a nontrivial generalization of a seminal result of Dinitz, Garg, and Goemans (1999) by showing how to efficiently turn a given transshipment~$x$ into an unsplittable transshipment~$y$ with~$y_a<x_a+d_{\max}$ for all arcs~$a$, where~$d_{\max}$ is the maximum demand (or supply) value. Further results include bounds on the number of rounds required to satisfy all demands, where each round consists of an unsplittable transshipment that routes a subset of the demands while respecting arc capacity constraints.
\end{abstract}

\setcounter{page}{1}

%----------
% Section 1
%----------

\section{Introduction}
\label{sec:introduction}

Network flows constitute one of the most fundamental classes of problems in combinatorial optimization and mathematical programming. We refer to the historical book by Ford and Fulkerson~\cite{FordFulkerson62}, the classical textbook by Ahuja, Magnanti, and Orlin~\cite{Ahuja:1993:NFT:137406}, as well as the more recent textbook by Williamson~\cite{williamson2019network} for comprehensive treatments of the subject.

Starting with the pioneering work of Kleinberg~\cite{Kleinberg_PhD,kleinberg1996single}, the past three decades have witnessed substantial progress in the study of unsplittable flows, in which the demand of each commodity must be routed along a single path from its source to its sink subject to given arc capacity constraints. Unsplittable flows are crucial in applications such as logistics, traffic, and telecommunication networks, where splitting a commodity across multiple paths may degrade quality of service, increase operational complexity, or be infeasible altogether, as exemplified by optical networks requiring specialized equipment for path splitting and freight logistics where dividing loads across routes is often impractical.

The most general multicommodity version of the problem, in which each commodity has its own source and sink, generalizes the arc-disjoint paths problem and is notoriously difficult to approximate. If the maximum demand is smaller than the minimum arc capacity, the randomized rounding technique of Raghavan and Thompson~\cite{Raghavan1987} achieves a logarithmic approximation factor for minimizing congestion, while for densely embedded unit-capacity graphs, Kleinberg~\cite{Kleinberg_PhD} obtains a constant-factor approximation with high probability. We refer to the survey by Kolliopoulos~\cite{Kolliopoulos_survey} for an overview of results on general unsplittable flows. The difficulty of the general problem has motivated extensive study of the special single source case, in which all commodities share a common source. In this setting, the underlying problem of computing a fractional (i.e., not necessarily unsplittable) flow that satisfies all demands reduces to a classical single commodity flow problem and can be solved via a maximum flow computation.

\paragraph{Unsplittable transshipments.}
We study unsplittable transshipments as an extension of single source unsplittable flows to networks with multiple sources and sinks and prescribed supplies and demands, where for each source--sink pair flow is routed along at most one path. More precisely, we are given a directed graph~$D = (V,A)$ together with a vertex balance function~$b\colon V \to \mathbb{R}$ satisfying~$\sum_{v \in V} b(v) = 0$. A vertex~$v$ is called a \emph{source} if~$b(v) > 0$, and the set of all sources is denoted by~$S^+$. Similarly, a vertex~$v$ is called a \emph{sink} if~$b(v) < 0$, and the set of all sinks is denoted by~$S^-$. A flow~$x\in\mathbb{R}^A_{\geq0}$ is a \emph{$b$-transshipment} if
\[
\sum_{a \in \delta^{\text{out}}(v)}x_a- \sum_{a \in \delta^{\text{in}}(v)} x_a = b(v) \qquad \text{for all~$v \in V$,}
\]
where~$\delta^{\text{in}}(v)$ and~$\delta^{\text{out}}(v)$ denote the sets of incoming and outgoing arcs of vertex~$v$, respectively. Given arc capacities~$c\in\mathbb{R}_{\geq 0}^A$, a~$b$-transshipment~$x$ is \emph{feasible} if~$x_a \leq c_a$ for all~$a \in A$. As for the special single source case, finding a feasible~$b$-transshipment reduces to a maximum flow computation. An \emph{unsplittable~$b$-transshipment} is a~$b$-transshipment together with a path decomposition consisting of source--sink paths only, and with at most one path for each source--sink pair. In order to emphasize the fact that a particular transshipment~$x$ is not necessarily unsplittable, we sometimes refer to~$x$ as a \emph{fractional} transshipment.

\paragraph{Single source unsplittable flows.}
The special case of unsplittable transshipments with a single source, that is,~$S^+ = \{s\}$, coincides with the well-studied Single Source Unsplittable Flows (SSUF) problem. The rich literature on SSUF mostly refers to the demand values at sinks as~$d_t\coloneqq |b(t)|$ for~$t\in S^-$. Arguably the most influential work on SSUF is due to Dinitz, Garg, and Goemans~\cite{DGG}. Their central result shows that any fractional flow~$x$ satisfying all demands can be efficiently transformed into a single source unsplittable flow~$y$ such that
\begin{equation*}
    y_a < x_a + d_{\max} \qquad \text{for all arcs } a \in A,
\end{equation*}
where~$d_{\max} \coloneqq \max_{t \in S^-} d_t$ denotes the maximum demand of any sink. In particular, in a capacitated single source network in which all arc capacities are at least~$d_{\max}$ (so that any commodity can be routed along any arc) this result immediately yields a~$2$-approximation algorithm for finding an unsplittable flow that minimizes congestion. The result has also inspired a substantial body of subsequent research, ranging from contributions by Kolliopoulos and Stein~\cite{Kolliopoulos_3_2_SSUF} and Skutella~\cite{Skutella_Apprx_SSUF_3_1} in the early 2000s, to recent work by Morell and Skutella~\cite{MoSk_2022}, Traub, Vargas Koch, and Zenklusen~\cite{traub2024single}, Majthoub Almoghrabi, Skutella, and Warode~\cite{majthoub2025integer}, as well as Swamy, Traub, Vargas Koch, and Zenklusen~\cite{swamy2026unsplittable}.

Beyond the congestion-minimization problem discussed above, Kleinberg~\cite{kleinberg1996single} introduces further variants of unsplittable flow problems. The \emph{minimum number of rounds} problem asks for a partition of the set of sinks into a minimum number of subsets (rounds), together with a feasible unsplittable flow for each subset. Building upon their main result, Dinitz, Garg, and Goemans~\cite{DGG} obtain a~$5$-approximation algorithm for this problem, improving upon earlier approximation results by Kleinberg~\cite{kleinberg1996single} and Kolliopoulos and Stein~\cite{Kolliopoulos_3_2_SSUF}. In particular, if a feasible fractional flow satisfying all demands exists, then all commodities can be routed unsplittably in five rounds, and there are instances for which three rounds are necessary~\cite{DGG}.

Another variant of unsplittable flows is the \emph{maximum routable demand} problem, which seeks a feasible unsplittable flow for a subset of demands that maximizes the total routed demand. Again improving upon earlier results by Kleinberg~\cite{kleinberg1996single} and Kolliopoulos and Stein~\cite{Kolliopoulos_3_2_SSUF}, Dinitz, Garg, and Goemans~\cite{DGG} present a~$0.226$-approximation algorithm for this problem. They also exhibit an instance that admits a feasible fractional solution but for which at most a~$0.385$ fraction of the total demand can be routed unsplittably.

Beyond the single source case and the results discussed above, unsplittable flows have been studied extensively in more general settings and under a variety of objective functions. For a broader overview of this literature, we refer to the survey by Kolliopoulos~\cite{Kolliopoulos_survey} and to the work of Grandoni, Mömke, and Wiese~\cite{grandoni2022ptas}.

\paragraph{Contribution and outline.}
We introduce the \emph{Unsplittable Transshipment Problem}, a natural generalization of SSUF to networks with multiple sources and sinks and prescribed supplies and demands. Our main contribution is a nontrivial extension of the seminal result of Dinitz, Garg, and Goemans~\cite{DGG}. We show how to efficiently transform any feasible fractional transshipment into an unsplittable transshipment while increasing the flow on each arc by at most the maximum demand. Compared to SSUF, unsplittable transshipments are more challenging since, in addition to selecting suitable paths for source--sink pairs, one must also determine appropriate flow values on these paths so that both supplies and demands are satisfied exactly. For example, even in networks with unit arc capacities and integral supplies and demands, deciding whether a feasible unsplittable transshipment exists is \textsf{NP}-hard; see Appendix~\ref{app:hardness}. Moreover, even when all supplies, demands, and arc capacities are integral, feasible unsplittable transshipments need not be integral; see Appendix~\ref{app:non-integrality}.

Our main algorithm builds upon the DGG Algorithm and augments it with a carefully designed mechanism for splitting sinks and their demands into sub-sinks that are routed to distinct sources, thereby ensuring feasibility of the resulting unsplittable transshipment. Moreover, the solutions produced by our algorithm exhibit strong structural properties: the induced source--sink bipartite graph, where a source~$s$ is connected to a sink~$t$ if the unsplittable transshipment routes flow along an~$s$--$t$ path, is acyclic, yielding a tree-like structure. Moreover, the flow entering each sink~$t$ is confluent, meaning that any two paths ending at~$t$ that meet at some vertex coincide from that point onward until they reach~$t$.

Beyond congestion, we also investigate the \emph{minimum number of rounds} and the \emph{maximum routable demand} objectives. Also these variants turn out to be more intricate for unsplittable transshipments than in the single source setting. In particular, we derive constant bounds on the number of rounds only under additional structural assumptions, and we obtain such results only in special cases.

The paper is organized as follows. In Section~\ref{sec:original_DGG}, we review the DGG Algorithm for single source unsplittable flows and highlight the key ideas underlying its analysis. Section~\ref{sec:modified_DGG} presents our Modified DGG Algorithm for unsplittable transshipments and proves the main structural and approximation results. In Section~\ref{sec:beyond_congestion}, we investigate routing in rounds and the maximum routable demand problem under additional assumptions. We conclude in Section~\ref{sec:conclusion} by highlighting several open questions for future research.

%----------
% Section 2
%----------

\section{A short review of the Dinitz-Garg-Goemans Algorithm for SSUFs}
\label{sec:original_DGG}

In the following, we summarize the algorithm of Dinitz, Garg, and Goemans~\cite{DGG} (the \emph{DGG Algorithm}) and highlight key ideas underlying its analysis. Subsequently, in Section~\ref{sec:modified_DGG}, we present an extended version of the DGG Algorithm for the Unsplittable Transshipment Problem.

\subsection{Description of the DGG Algorithm}

We consider an instance of the SSUF problem on an acyclic digraph~$D = (V,A)$ with a single source~$s$, a set of sinks~$S^-$ with associated demands~$d_t$ for~$t \in S^-$, and a flow~$x$ that satisfies these demands. The DGG Algorithm iteratively modifies the initial flow~$x$, and whenever the flow on an arc is reduced to zero during this process, the arc is removed from the network.

In a preliminary phase, starting from the given flow~$y\coloneqq x$, the DGG~Algorithm iteratively moves sinks backward toward the source~$s$ along flow-carrying arcs whose flow meets or exceeds the sink’s demand, decreasing the flow~$y$ on the traversed arcs accordingly. The preliminary phase terminates once every sink has either reached the source~$s$ or is located at a vertex with at least two incoming arcs and is regular: A sink is called \emph{regular} if the flow on each arc entering its current vertex is strictly smaller than the demand of the sink; otherwise, the sink is \emph{irregular}.

Once the preliminary phase ends, the algorithm iteratively augments the current flow~$y$ along a carefully chosen alternating cycle and, in each iteration, tries to move the sinks toward the source~$s$ according to certain rules. 

A key notion for constructing alternating cycles is that of \emph{singular arcs}. At the beginning of an iteration, an arc~$(u,v)$ is labeled \emph{singular} if~$v$ and all vertices reachable from~$v$ have out-degree at most one. An alternating cycle is constructed as follows. Starting from an arbitrary vertex, follow the outgoing arcs until a \emph{junction vertex}~$w$, that is, a vertex  with no outgoing arcs, is reached (the \emph{forward path}). Then, beginning with an incoming arc at~$w$ that is distinct from the arc used to reach~$w$, construct a \emph{backward path} traversing the incoming singular arcs as far as possible. Upon reaching a vertex with at least two outgoing arcs, resume the construction of a forward path. This alternation of forward and backward paths is repeated until a vertex is revisited, thereby forming an \emph{alternating cycle}~$C$ consisting of alternating forward and backward paths.

The flow~$y$ is then augmented along the alternating cycle~$C$ as follows. The flow on each forward arc of~$C$ is decreased, and the flow on each backward arc of~$C$ is increased, by the same positive amount which is the minimum of two values: The first value is the minimum flow along a forward arc of~$C$. The second is the minimum of~$d_t-y_a$ taken over all backward arcs~$a=(v,w)$ of~$C$ and over all sinks~$t$ at~$w$ with~$d_t>y_a$. If the minimum is achieved by the flow~$y_a$ on a forward arc~$a$, then the augmentation reduces the flow on this arc to zero and the arc is thus deleted. Finally, if possible, sinks are moved backward towards source~$s$, preferably along singular arcs carrying exactly their demand, and otherwise along non-singular arcs carrying at least their demand.

The algorithm terminates once all sinks have reached the source~$s$, and the resulting unsplittable flow is defined by the collection of paths traced by the sinks during their backward movement.

\subsection{Analysis of the DGG Algorithm}

The correctness and analysis of the DGG Algorithm relies on the following invariants, maintained throughout its execution:
        
\begin{enumerate}[(i)]\itemsep0ex
\item\label{inv:flow_satisfies_demands} 
Flow~$y$ satisfies all current demands of sinks not yet moved to~$s$.

\vspace{-2ex}
\begin{proof}
In the augmentation step, flow is increased along backward paths and decreased along forward paths by the same amount, preserving the excess (inflow minus outflow) at each vertex, thus maintaining a flow that meets all current demands.    
\end{proof}
\vspace{-2ex}
\item\label{inv:vertex_w_sinks_2_inc}
At the beginning of each iteration, every vertex contains at most one irregular sink. If a vertex does contain an irregular sink, then its out-degree is zero and it also contains a regular sink. In particular, every vertex that contains sinks has at least two incoming arcs.

\vspace{-2ex}
\begin{proof}
We refer to the proof of Lemma 3.2 in~\cite{DGG}. 
\end{proof}

\vspace{-2ex}
Notice that the last part of this invariant is crucial for the construction of alternating cycles since it guarantees that, for each junction vertex, there exists an incoming arc different from the forward arc used to reach it such that a backward path can be constructed.
\item\label{inv:moving_removes_sing}
A singular arc is removed at the end of an iteration in which any sink moves along it.

\vspace{-2ex}
\begin{proof}
A sink moves along a singular arc in an iteration only when the flow on that arc exactly matches the demand of that sink before moving.
\end{proof}
\vspace{-2ex}
\end{enumerate}

Until an arc~$a$ becomes singular, its flow never increases and the total demand of sinks moved along it is bounded by its initial flow value~$x_a$. Once~$a$ is singular, at most one sink with demand at most~$d_{\max}$ can move along it before~$a$ disappears. Thus, in the final unsplittable flow, the flow on any arc is less than~$x_a + d_{\max}$. We refer the interested reader to~\cite{DGG} for further details.    
      
%----------
% Section 3
%----------

\section{Modified DGG Algorithm for unsplittable transshipments}
\label{sec:modified_DGG}

In this section, we modify the DGG Algorithm and its analysis in order to compute unsplittable transshipments with bounded congestion. We are given an acyclic digraph~$D = (V,A)$ with a set of sources~$S^+$ and a set of sinks~$S^-$, together with supplies and demands specified by a balance function $b \colon V \to \mathbb{R}$ and a fractional $b$-transshipment~$x$. We first transform this unsplittable transshipment instance into an SSUF instance by introducing a super-source~$s^*$ and dummy arcs $(s^*,s)$ to each source~$s \in S^+$, each carrying flow~$x_{(s^*,s)} \coloneqq b(s)$, and by setting~$d_t \coloneqq |b(t)|$ for all sinks~$t \in S^-$.

However, applying the DGG Algorithm to the resulting SSUF instance does not necessarily yield a valid $b$-transshipment. During the execution of the algorithm, the flow on dummy arcs~$(s^*,s)$ that lie on alternating cycles may be modified, which can result in an unsplittable flow that no longer satisfies the fixed supply constraints at the sources~$s \in S^+$.

On the other hand, the DGG Algorithm routes the entire demand of each sink along a single path and does not exploit the possibility of sending flow into a sink along multiple paths originating from different sources in~$S^+$. The central idea of our Modified DGG Algorithm is to leverage this flexibility in order to keep the flow on dummy arcs~$(s^*,s)$, and hence the supplies at the sources, fixed.

\subsection{Description of the Modified DGG Algorithm}

As in the original DGG Algorithm, the Modified DGG Algorithm begins with the flow~$y \coloneqq  x$ and, in a preliminary phase, iteratively moves sinks backward toward the super-source~$s^*$ along flow-carrying arcs incident to sinks whose flow meets or exceeds their demands. After the preliminary phase, the algorithm iteratively attempts to construct alternating cycles that do not contain the super-source~$s^*$. Augmenting flow along a cycle that includes~$s^*$ would modify the flow on arcs~$(s^*,s)$ and thus alter the fixed supply values of such sources~$s\in S^+$. We therefore proceed as follows.

\paragraph{Nice alternating cycles and singular digraphs.}

As in the original DGG Algorithm, we attempt to construct an alternating cycle by alternating between \emph{forward} (arbitrary) and \emph{backward} (singular) arcs. The last vertex reached on a forward path during an iteration is a junction vertex~$w$ with no outgoing arcs. All incoming arcs of~$w$ are singular, and, by flow conservation,~$w$ necessarily contains one or more sinks. Furthermore, a vertex is called a \emph{funnel vertex} if it has at most one outgoing arc; otherwise, it is called a \emph{non-funnel vertex}. Note that the head of any singular arc is a funnel vertex, and that backward paths in alternating cycles always terminate upon reaching a non-funnel vertex. We call an alternating cycle \emph{nice} if it does not contain the super-source~$s^*$. Nice alternating cycles allow us to augment the flow without changing the supplies at the sources in~$S^+$ and thus to proceed exactly as in the original DGG Algorithm.

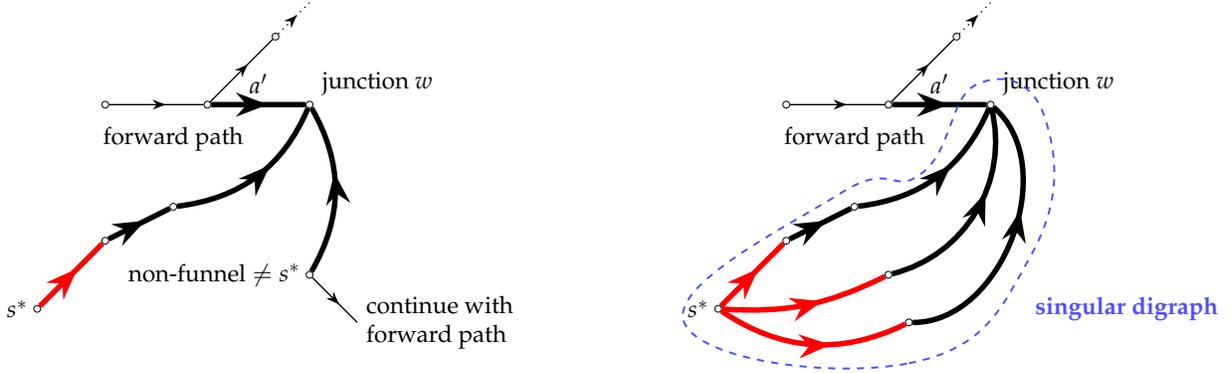
\begin{figure}[t]
    \centering
    \tikzset{
    midarrow/.style={postaction={decorate, decoration={markings, mark=at position 0.6 with {\arrow{Stealth}}}}}
}

\tikzset{
  >=Stealth,
  v/.style={circle,draw,inner sep=1pt},
  lab/.style={draw=none,font=\small},
  fwd/.style={line width=0.2mm, midarrow},
  bwd/.style={line width=0.8mm, midarrow},
  sing/.style={line width=0.8mm, midarrow},
  junction/.style={circle,draw,very thick,inner sep=1pt},
  nonfunnel/.style={circle,draw,thick,double},
}

\begin{tikzpicture}

% ================= LEFT DIAGRAM =================
\begin{scope}

% Nodes
\node[v] (v)  at (0,5) {};
\node[v] (f1) at (1.5,5) {};
\node[v] (f2) at (2.5,6) {};
\node[v] (w) at (3,5) {};
\node[v] (nf) at (1,3.5) {};
\node[v] (x)  at (0,3) {};
\node[v] (s) at (-1,2) {};
\node[v] (nf1) at (3,2.5) {};

% Labels
%\node[lab,left] at (v) {$v$};
\node[lab,above right] at (w.east) {junction $w$};
\node[lab,left] at (s) {$s^*$};
\node[lab] at (1,4.5) {forward path};
\node[lab,left,align=left] at (nf1) {non-funnel $\neq s^*$};
\node[lab,right,align=left] at (3.7,1.8) {continue with\\ forward path};

% Forward path
\draw[fwd] (v)--(f1);
\draw[fwd] (f1)--(f2);
\draw[fwd, dotted] (f2)--(3,6.5);
\draw[sing] (f1)--node[lab,above]{$a'$}(w);

% Backward path to s*
\draw[bwd] (nf) to[bend right=30] (w);
\draw[sing] (x)--(nf);
\draw[sing,red] (s)--(x);

% Alternative continuation
\draw[bwd] (nf1) to[bend right=30] (w);
\draw[fwd] (nf1)--(3.7,1.8);

\end{scope}

%\node[lab] at (7,4) {\textbf{OR}};

% ================= RIGHT DIAGRAM =================

\begin{scope}[shift={(10,0)}]

% Nodes
\node[v] (v)  at (0,5) {};
\node[v] (f1) at (1.5,5) {};
\node[v] (f2) at (2.5,6) {};
\node[v] (w) at (3,5) {};
\node[v] (nf) at (1,3.5) {};
\node[v] (x)  at (0,3) {};
\node[v] (s) at (-1,2) {};
\node[v] (a)  at (1.5,2.5) {};
\node[v] (b)  at (1.8,1.8) {};

% Labels
%\node[lab,left] at (v) {$v$};
\node[lab,above right] at (w.east) {junction $w$};
\node[lab,left] at (s) {$s^*$};
\node[lab] at (1,4.5) {forward path};

% Forward path
\draw[fwd] (v)--(f1);
\draw[fwd] (f1)--(f2);
\draw[fwd, dotted] (f2)--(3,6.5);
\draw[sing] (f1)--node[lab,above]{$a'$}(w);

% Multiple backward paths
\draw[bwd] (nf) to[bend right=30] (w);
\draw[sing] (x)--(nf);
\draw[sing,red] (s)--(x);

\draw[bwd] (a) to[bend right=45] (w);
\draw[bwd,red] (s) to[bend right=15] (a);

\draw[bwd] (b) to[bend right=70] (w);
\draw[bwd,red] (s) to[bend right=30] (b);

% Singular digraph boundary
\draw[thick,dashed,blue!70,smooth cycle]
  plot[tension=0.8] coordinates {
    (-1.5,1.8)
    (1.8,1.2)
    (3.6,2.8)
    (3.8,4.8)
    (2.8,5.3)
    (2,4)
    (0.8,3.7)
  };

\node[lab,align=center, blue!70] at (5,2) {\textbf{singular digraph}};

\end{scope}

\end{tikzpicture}
    \vspace{-.5cm}
    \caption{Backward path discovery in the Modified DGG Algorithm starting at junction vertex~$w$ with singular arcs depicted in thick: \emph{Left}: a non-funnel vertex~$\neq s^*$ is found along some backward path, allowing to continue the construction of a nice alternating cycle; \emph{Right}: the only non-funnel vertex found on backward paths is~$s^*$, yielding a singular digraph rooted at junction vertex~$w$.}
    \label{fig:try_nice_alt}      
\end{figure}  

There are, however, situations in which no nice alternating cycle can be found. This occurs precisely when, starting from a junction vertex~$w$, all possible backward paths lead to the super-source~$s^*$. Equivalently,~$s^*$ is the only non-funnel vertex reachable from~$w$ via backward paths. This can be verified by a simple backward depth-first search (DFS) starting at~$w$ while ignoring the arc~$a'$ used to reach~$w$. If backward DFS finds a non-funnel vertex~$v\neq s^*$, we can continue our search for a nice alternating cycle along a forward path from~$v$. Otherwise, if all backward paths lead to~$s^*$,  we refer to the sub-digraph explored by this backward DFS as the \emph{singular digraph rooted at~$w$}.  

If a singular digraph rooted at the junction vertex~$w$ is found, we select a (sub-)sink~$t$ located at~$w$ with~$d_t > y_{a'}$. The existence of such a (sub-)sink is guaranteed by Invariant~\ref{inv:vertex_w_sinks_2_inc}, whose validity is established below. This (sub-)sink~$t$ is uniquely chosen whenever a singular digraph is encountered and will later be split into two sub-sinks~$t_1$ and~$t_2$ with demands~$d_{t_1} \coloneqq y_{a'}$ and $d_{t_2} \coloneqq d_t - y_{a'}$, respectively. After all other (sub-)sinks in the singular digraph have been moved to the super-source as described below, the sub-sink~$t_2$ is created and routed to the super-source as well. Subsequently, the sub-sink~$t_1$ is moved backward along the arc~$a'$, reducing the flow on~$a'$ to zero and thereby removing~$a'$ from the network.

\begin{lemma}
\label{lem:singular-digraph}
With the exception of arc~$a'$, all incoming arcs of vertices in the singular digraph rooted at~$w$ belong to the singular digraph. Moreover, every vertex~$v\neq s^*$ in the singular digraph has at most one outgoing arc, and this outgoing arc belongs to the singular digraph.
\end{lemma}

\begin{proof}
The first statement follows from the fact that the singular digraph is constructed via a backward depth-first search. For the second statement, observe that every vertex~$v\neq s^*$ in the singular digraph is a funnel vertex and hence has at most one outgoing arc. Moreover, this outgoing arc lies on the unique path from~$v$ to the junction vertex~$w$, which is entirely contained in the singular digraph.  
\end{proof}

In other words, Lemma~\ref{lem:singular-digraph} states that the sub-digraph of the current network induced by all vertices of the singular digraph except~$s^*$ is an in-tree, i.e., a directed tree rooted at~$w$ with all arcs oriented toward it.
As a consequence of Lemma~\ref{lem:singular-digraph}, the flow~$y$ on the singular digraph exactly satisfies the demands of all (sub-)sinks located at its vertices, with the sole exception of the uniquely chosen (sub-)sink~$t$, for which only the amount~$d_{t_2} \coloneqq d_t - y_{a'}$ is satisfied within the singular digraph. Since the singular digraph contains at most one path from each source in~$S^+$ to each of its (sub-)sinks, any path decomposition of its flow yields an unsplittable transshipment that satisfies all these demands.

Equivalently, the algorithm decomposes the flow~$y$ on the singular digraph into paths, splits each (sub-)sink into as many sub-sinks as there are paths serving it, and then moves each such sub-sink backward along its corresponding path to the super-source~$s^*$. Finally, all arcs of the singular digraph are removed.

At the end of an iteration, further sinks may move backward towards super-source~$s^*$ exactly as in the original DGG Algorithm. Algorithm~\ref{alg:modifiedDGG} gives a summary of the Modified DGG Algorithm. 

\begin{algorithm}[t]
\caption{Modified DGG Algorithm\label{alg:modifiedDGG}}
\begin{minipage}{\linewidth}
\hspace*{\algorithmicindent}\textbf{Input:} Unsplittable transshipment instance~$(D, S^+, S^-, b, x)$ \\
\hspace*{\algorithmicindent}\textbf{Output:} Unsplittable transshipment defined by path sets~$\{\mathcal{P}_t\}_{t\in S^-}$
\begin{algorithmic}[1]
\State Construct corresponding SSUF instance~$(D^*, s^*, S^-, d, x)$; set~$y\coloneqq x$
\While{there exists~$a=(v,t)\in A$ with~$y_a\geq d_t$} \Comment{\textbf{Preliminary phase}}
    \State decrease~$y_a$ by~$d_t$; move~$t$ to~$v$; delete~$a$ if~$y_a=0$
\EndWhile
%\State Remove sinks that reached~$s^*$
\While{some (sub-)sink has not reached~$s^*$}\Comment{\textbf{Iterations}}
    \State Label singular arcs and funnel vertices
    \If{a \textit{nice alternating cycle} is found}
        \State Augment flow and move (sub-)sinks exactly as in the original DGG Algorithm
    \Else
        \State Move the (sub-)sinks in the singular digraph \hyperlink{alg:singular-footnote}{\textsuperscript{$1$}} backward to the super-source
        \label{step:singular}
        \State Move further (sub-)sinks according to DGG moving rules
    \EndIf
\EndWhile
%\State \Return~$\{\mathcal{P}_i\}_{i\in S^-}$
\end{algorithmic}

% \footnotetext{`\emph{(sub-)sinks in the singular digraph}' refers to all (sub-)sinks other than~$t$ located at vertices of the singular digraph, and the sub-sink~${t_2}$ of~$t$.}
\medskip
\footnotesize
\hypertarget{alg:singular-footnote}{\textsuperscript{$1$}}
\emph{(sub-)sinks in the singular digraph} refers to all (sub-)sinks other than~$t$
located at vertices of the singular digraph, and the sub-sink~$t_2$ of~$t$.
\end{minipage}

\end{algorithm}

\subsection{Analysis of the Modified DGG Algorithm}

The correctness of the Modified DGG Algorithm is based on the preservation of the same Invariants~\ref{inv:flow_satisfies_demands},~\ref{inv:vertex_w_sinks_2_inc}, and~\ref{inv:moving_removes_sing} as the original DGG Algorithm.

\begin{proof}[Proof of Invariant~\ref{inv:flow_satisfies_demands}]
If a nice alternating cycle is found in an iteration, the flow is augmented using the same procedure as in the original DGG Algorithm. This augmentation preserves the property that the flow satisfies the demands of all current (sub-)sinks that have not yet reached~$s^*$. In the remaining iterations, the flow on the singular digraph exactly satisfies the demands of the \mbox{(sub-)sinks} located at its vertices. Consequently, removing the arcs of the singular digraph together with its (sub-)sinks does not violate Invariant~\ref{inv:flow_satisfies_demands}.
\end{proof}

\begin{proof}[Proof of Invariant~\ref{inv:vertex_w_sinks_2_inc}]
The proof proceeds by induction on the number of iterations. All sinks are regular at the end of the preliminary phase such that the invariant holds at the beginning of the first iteration. If a nice alternating cycle is found in an iteration, Invariant~\ref{inv:vertex_w_sinks_2_inc} is preserved exactly as shown in the proof of~\cite[Lemma~3.2]{DGG}. Otherwise, the algorithm identifies a singular digraph rooted at the junction vertex~$w$ and splits a sink~$t$ located at~$w$ into two sub-sinks~$t_1$ and~$t_2$ such that the demand~$d_{t_1}$ of~$t_1$ equals the flow on the arc~$a'$ along which the forward path reaches~$w$. Except for~$t_1$, all (sub-)sinks at the vertices of the singular digraph are deleted together with its arcs, while~$t_1$ is moved to the tail of~$a'$ (and possibly further), and the arc~$a'$ is removed.

Consequently,~$t_1$ is the only sink that can potentially violate Invariant~\ref{inv:vertex_w_sinks_2_inc}. Suppose that~$t_1$ is eventually moved to a vertex~$v$. If~$t_1$ is regular at~$v$, the invariant is clearly preserved. Otherwise,~$t_1$ is irregular at~$v$, which implies the existence of an incoming arc~$a$ with flow value~$y_a \ge d_{t_1}$. Since~$t_1$ does not move further backward along~$a$, it must hold that~$y_a > d_{t_1}$ and that the arc~$a$ was labeled singular at the beginning of the iteration.

In particular, the outgoing arc of~$v$ along which~$t_1$ was moved backward to~$v$ must have been singular and, moreover, the only outgoing arc of~$v$. This arc is therefore deleted after the move, implying that the new out-degree of~$v$ is zero. Since there was no irregular sink at~$v$ at the beginning of the iteration (as its out-degree was one),~$t_1$ is the only irregular sink at~$v$. Finally, as~$y_a > d_{t_1}$, there must exist at least one additional sink at~$v$, which completes the proof.
\end{proof}

By Invariant~\ref{inv:vertex_w_sinks_2_inc}, whenever a forward path reaches a junction vertex~$w$ along an arc~$a'$, there exists at least one additional incoming (singular) arc at~$w$ along which the backward depth-first search can continue and either reach a non-funnel vertex~$v \neq s^*$ or identify a singular digraph. Consequently, in every iteration, the algorithm finds either a nice alternating cycle or a singular digraph, as described in the algorithm.

\begin{proof}[Proof of Invariant~\ref{inv:moving_removes_sing}]
A (sub-)sink can be moved along a singular arc in two distinct ways. If it is moved according to the first DGG moving rule, the flow on the arc is reduced to zero and the arc is removed immediately. Otherwise, the arc belongs to a singular digraph, along which (sub-)sinks are routed without augmenting the flow; all arcs in this singular digraph are removed at the end of the iteration. In both cases, any singular arc used to move a (sub-)sink toward the source is removed from the digraph within the same iteration, thereby preserving the invariant.
\end{proof}

Next we prove that the~$b$-transshipment computed by the Modified DGG Algorithm is unsplittable. We actually prove the even stronger property that the demand of every sink is routed \emph{confluently}; that is, the arcs used to route the demand of any sink form an in-tree, i.e., a directed tree rooted at the sink with all arcs oriented toward it.

\begin{lemma}
\label{lem:confluent}
The Modified DGG Algorithm computes a~$b$-transshipment in which the demand of every sink is routed confluently.
\end{lemma}

\begin{proof}
We argue that, after deleting the super-source~$s^*$, the paths along which a sink~$t$ (or its sub-sinks) are moved backward toward~$s^*$ form an arborescence (out-tree) rooted at~$t$, whose leaves are sources in~$S^+$. Equivalently, the arcs used to route the demand of~$t$ form an in-tree.

During the preliminary phase, as well as in later iterations in which a (sub-)sink~$t$ is not located at a vertex of a singular digraph, it is moved backward along a single path toward the super-source~$s^*$. Therefore, it remains to consider iterations in which~$t$ is located at a vertex of a singular digraph rooted at a junction vertex~$w$. We distinguish two cases.

\emph{First case.} 
Assume that~$t$ is the uniquely chosen (sub-)sink located at~$w$ that is split by the algorithm into two sub-sinks~$t_1$ and~$t_2$. In this case,~$t_2$ may be further split into sub-sinks that are routed backward toward~$s^*$ within the singular digraph. By Lemma~\ref{lem:singular-digraph}, the corresponding backward paths, after deleting~$s^*$, form a tree. Since all arcs of the singular digraph are removed at the end of the iteration, all its vertices except for~$s^*$ have no outgoing arcs afterwards. Consequently, the paths along which the sub-sink~$t_1$ (or its sub-sinks) is later moved backward toward~$s^*$ do not intersect any of these vertices.

\emph{Second case.}  
Otherwise, by the same argument as for~$t_2$ in the first case, the backward paths taken by~$t$ (or its sub-sinks), after deleting~$s^*$, form a tree.

In summary, the backward paths along which the demand of a sink~$t$ (or its sub-sinks) is moved form, after deleting the super-source~$s^*$, an arborescence rooted at~$t$.
\end{proof}

We are now ready to prove the main result of this section, generalizing the classical result of Dinitz, Garg, and Goemans~\cite{DGG} to the Unsplittable Transshipment Problem.

\begin{theorem}\label{thm:main_result_UT}
Given a~$b$-transshipment~$x$, the Modified DGG Algorithm efficiently computes an unsplittable transshipment such that the flow on each arc~$a$ is strictly less than~$x_a + d_{\max}$.
\end{theorem}
         
\begin{proof}
Lemma~\ref{lem:confluent} implies that the algorithm computes an unsplittable transshipment. Thus, it remains to prove the bound on the flow values.
As in the original DGG Algorithm, the flow on a non-singular arc is never increased. Hence, the total demand routed along an arc up to the point at which it becomes singular is less than its initial flow. Subsequently, the arc~$a$ may appear repeatedly as a backward arc in nice alternating cycles, during which its flow may be increased multiple times. Observe that~$a$ remains singular until it is removed from the network, since no arcs are added during the algorithm and vertex out-degrees never increase. In the following, we argue that the total increase of flow on~$a$ is bounded by~$d_{\max}$.

After the last increase of flow on~$a$, let~$y$ denote the resulting flow and let~$t$ be the first regular (sub-)sink encountered at a vertex~$v$ along the unique directed path starting at the head of~$a$. If~$v$ is the head of~$a$, set~$a_v \coloneqq  a$; otherwise, let~$a_v$ be the incoming arc of~$v$ on this path. By Invariant~\ref{inv:vertex_w_sinks_2_inc},~$v$ is the first vertex on this path that contains a sink. Hence, by flow conservation,~$y_a \leq y_{a_v} \leq d_t$. In particular, the total increase of flow on arc~$a$ is strictly less than~\mbox{$d_t\leq d_{\max}$}.
\end{proof}

The \emph{congestion} of a transshipment is defined as the maximum, over all arcs, of the ratio between the flow routed on the arc and its capacity. We assume the standard \emph{balance condition}, that is, the capacity of every arc is at least the maximum demand~$d_{\max}$. Under this assumption, Theorem~\ref{thm:main_result_UT} immediately implies the following corollary.

\begin{corollary}
Under the balance condition, the Modified DGG Algorithm yields a~$2$-approximation algorithm for minimizing congestion in unsplittable transshipment instances.
\end{corollary}

Instead of working with a super-source~$s^*$ as we did above, one could just as well introduce a super-sink~$t^*$ and then reverse all arc directions.

\begin{remark}
Reversing all arcs yields a maximum capacity violation of at most~$\max_{s\in S^+} b(s)$. Thus, the Modified DGG Algorithm computes, in polynomial time, an unsplittable transshipment whose arc capacity violation is bounded by the minimum of the maximum supply~$\max_{s \in S^+} b(s)$ and the maximum demand~$\max_{t \in S^-} |b(t)|$.
\end{remark}

\subsection{Further observations}

\paragraph{Arc-wise lower bounds for unsplittable transshipments.}

Rohwedder and Svensson (personal communication, 2021) observed that the DGG Algorithm can be adapted to guarantee lower bounds on arc flows for unsplittable flows. By applying the same idea, one obtains an analogous result for the Unsplittable Transshipment Problem.

\begin{theorem}\label{thm:lower_bounds}
Given a~$b$-transshipment~$x$, the adapted Modified DGG Algorithm efficiently computes an unsplittable transshipment such that the flow on each arc~$a$ is greater than~$x_a - d_{\max}$.
\end{theorem}

As with the original DGG Algorithm, all steps remain the same except the augmentation step when a nice alternating cycle is found. In that case, we augment in the opposite direction: we increase the flow on the forward paths and decrease it on the backward paths.

\paragraph{Tightness of the Modified DGG Algorithm.}
For the SSUF problem, Dinitz, Garg, and Goemans~\cite{DGG} show that a capacity violation of~$d_{\max}$ is unavoidable in the worst case. Since SSUF is a special case of the unsplittable transshipment problem, this bound is also tight for unsplittable transshipments when each sink is reachable from a single source. The tightness of the bound is less immediate in highly connected instances, where each sink is reachable from many sources, and its demand could, in principle, be split so as to reduce congestion. We show that this intuition is misleading: the bound of~$d_{\max}$ remains tight even when every sink is connected by directed paths to an arbitrarily large number of sources. A concrete construction establishing this result is given in Appendix~\ref{app:sec:tightness}.

\paragraph{Bounding the total number of paths.}
By definition, an arbitrary unsplittable transshipment may use up to~$|S^+|\cdot|S^-|$ distinct paths. However, the following theorem implies that any unsplittable transshipment produced by the Modified DGG Algorithm uses at most~$|S^+| + |S^-| - 1$ paths. We refer to Appendix~\ref{app:sec:number_of_paths} for a proof.

\begin{theorem}
\label{thm:tree}
Given an unsplittable transshipment, define a bipartite graph with vertex set~$S^+ \cup S^-$ by adding an edge between a source~$s \in S^+$ and a sink~$t \in S^-$ if and only if the transshipment routes flow along an~$s$--$t$ path. Then, the bipartite graph associated with any unsplittable transshipment computed by the Modified DGG Algorithm is acyclic, and thus contains at most~$|S^+| + |S^-| - 1$ edges.
\end{theorem}

\paragraph{Cost of Confluence.}
Lemma~\ref{lem:confluent} implies that the Modified DGG Algorithm returns a sink-wise confluent unsplittable transshipment. This leads to the question of whether such confluence has an inherent cost: are there instances with a feasible (non-confluent) unsplittable transshipment for which every confluent unsplittable transshipment has congestion close to~$d_{\max}$, i.e., the largest possible gap?

The answer is yes, as demonstrated by a family of instances parameterized by~$q$, illustrated for~$q=4$ in Figure~\ref{fig:cost_of_confluence}.
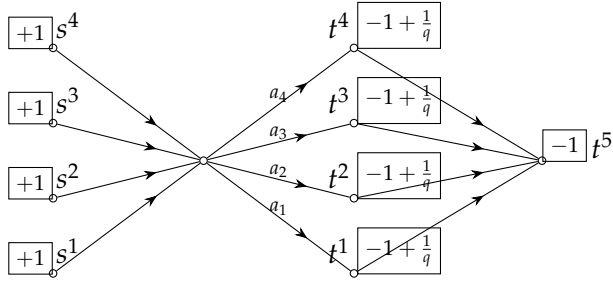
\begin{figure}[t] 
\centering
\tikzset{
    dot/.style={circle,draw,inner sep=1pt},
    lab/.style={draw=none,font=\small,inner sep=1pt},
    midarrow/.style={postaction={decorate,decoration={markings,mark=at position 0.7 with {\arrow{Stealth}}}}}
}

\begin{tikzpicture}

% -------------------- Nodes --------------------
% Left column s^i
\foreach \i [count=\y from 5] in {1,2,3,4} {
    \node[dot] (s\i) at (2,\y) {};
    \node[draw, scale=0.75] at (1.7,\y+0.2) {$+1$};
    \node at (2.2,\y+0.3) {$s^\i$};
}

% Right column t^i
\foreach \i [count=\y from 5] in {1,2,3,4} {
    \node[dot] (t\i) at (6,\y) {};
    \node[draw, scale=0.75] at (6.6,\y+0.3) {$-1+\frac{1}{q}$};
    \node at (5.8,\y+0.3) {$t^\i$};
}

% Center node
\node[dot] (c) at (4, 6.5) {};

% Far right node
\node[dot] (t5) at (8.5,6.5) {};
\node[draw, scale=0.75] at (8.8,6.7) {$-1$};
\node at (9.3,6.7) {$t^5$};

% -------------------- Edges --------------------
% Left -> center
\foreach \i in {1,2,3,4} {
    \draw[midarrow] (s\i) -- (c);
}

% Center -> right column with labels
\foreach \i/\a/\pos in {1/1/above,2/2/above,3/3/below,4/4/below} {
    \draw[midarrow] (c) -- node[lab,above,scale=0.75] {$a_\a$} (t\i);
}

% Right column -> far right
\foreach \i in {1,2,3,4} {
    \draw[midarrow] (t\i) -- (t5);
}

\end{tikzpicture}
\caption{Family of instances (depicted for~$q=4$) with unit arc capacities and~$d_{\max}=1$, admitting a feasible non-confluent unsplittable transshipment, while any confluent unsplittable transshipment has congestion~$1-\nicefrac1q$, tending to~$d_{\max}$ as~$q \to \infty$.}
\label{fig:cost_of_confluence}
\end{figure}
Each sink~$t^i$ with~$i\in \{1, \ldots, q\}$ must route its demand through arc~$a_i$. Thus, any unsplittable routing reserves~$1- \nicefrac{1}{q}$ units of capacity on~$a_i$ for sink~$t^i$. If we additionally require confluence for sink~$t^{q+1}$ (denoted~$t^5$ in Figure~\ref{fig:cost_of_confluence}), then some arc~$a_i$ must carry a total flow of~$1 + (1- \nicefrac{1}{q})$. Hence, in any confluent unsplittable transshipment, the capacity violation on that arc is at least~$1- \nicefrac{1}{q}$, which tends to~$d_{\max} = 1$ as~$q \to \infty$.

%----------
% Section 4
%----------

\section{Minimum number of rounds and maximum routable demand}
\label{sec:beyond_congestion}

In this section, let the digraph~$D=(V,A)$ have arc capacities~$c_a \geq 0$,~$a\in A$. We impose the \emph{balance condition} that~$d_{\max}\coloneqq \max_{t\in S^-}|b(t)| \leq \min_{a \in A} c_a =: c_{\min}$ and assume that a feasible fractional transshipment exists for the given unsplittable transshipment instance. Let~$d_{\tot}\coloneqq \sum_{t\in S^-}|b(t)|$ denote the total demand of all sinks, and let~$b^+_{\avg}\coloneqq \sum_{s \in S^+}\nicefrac{b(s)}{|S^+|}= \nicefrac{d_{\tot}}{|S^+|}$ denote the average supply of a source. Note that we make no assumptions about the supply values; in particular, sources and sinks need not be symmetric with respect to supply and demand values.

In the context of unsplittable transshipments, the \emph{minimum number of rounds} problem asks for a partition of the set of sinks into a minimum number of subsets (rounds), together with a feasible unsplittable transshipment for each subset, such that for every source the total flow sent to sinks over all rounds equals its prescribed supply. Moreover, the \emph{maximum routable demand} problem seeks a feasible unsplittable transshipment that satisfies the demands of a subset of sinks, maximizes the total demand of this subset, and does not exceed the available supplies at the sources.

\subsection{Minimizing the number of rounds}\label{subsec:number_of_rounds}

\paragraph{SSUF versus unsplittable transshipments.} Under the balance condition, it is shown in~\cite{DGG} that all demands in an SSUF instance can be routed in polynomial time using at most five rounds, by partitioning demands into \emph{small} and \emph{large}, and routing each class separately on suitably capacity-scaled copies of the network.

For unsplittable transshipments, however, routing in rounds poses additional challenges because rounds are \emph{not independent}: the supply of a source consumed in earlier rounds reduces the supply available in later rounds. This cross-round `memory' has no analogue in the SSUF setting. Moreover, unlike in SSUF, where all indivisible demands are known in advance, the Modified DGG Algorithm may generate arbitrarily small sub-demands during its execution, so that no meaningful global partition of demands into `small' and `large' classes exists. Furthermore, scaling down capacities does not guarantee that at most one demand uses a given arc. As a consequence, feasibility arguments based on capacity scaling, such as those used in~\cite{DGG}, no longer apply. Routing in rounds for unsplittable transshipments therefore requires more refined techniques.

For the case in which demand values are sufficiently smaller than~$c_{\min}$, we prove bounds on the number of rounds for several such regimes in Appendix~\ref{app:finite-rounds_routing}. In the following, we present a unified strategy that routes all instances with~$d_{\max} < c_{\min}$ in a number of rounds that depends only on the gap between~$d_{\max}$ and~$c_{\min}$.

\paragraph{Routing demands in rounds \texorpdfstring{when~$d_{\max}<c_{\min}$}{}.}

We may assume without loss of generality that~$d_{\max}\leq (1-\nicefrac{1}{n})c_{\min}$ for some~$n\in\mathbb{N}$. Following the construction of~\cite[Fig.~6]{DGG}, we construct an auxiliary digraph~$D'=(V',A')$ consisting of~$n+1$ identical copies of~$D$, where each copy~$a'$ of an arc~$a\in A$ has capacity~$c'_{a'} \coloneqq \nicefrac{c_a}{(n+1)}$. For each sink~$t \in S^-$, we add a super-sink~$T$ and connect every copy of~$t$ to~$T$ with an arc of capacity~$\nicefrac{d_t}{(n+1)}$. Symmetrically, for each source~$s \in S^+$ we introduce a head-source~$S$ and connect every copy of~$s$ to~$S$ with an arc of capacity~$\nicefrac{b(s)}{(n+1)}$, and finally connect every head-source~$S$ to a global super-source~$S^*$ by an arc of capacity~$b(s)$. Routing a fraction~$\nicefrac{1}{(n+1)}$ of the original feasible flow through every copy yields a single source flow from~$S^*$ to the super-sinks in~$D'$, to which we apply the Modified DGG Algorithm.

Since~$c'_{a'} + d_{\max} \leq \nicefrac{c_a}{(n+1)} + \left(1-\nicefrac{1}{n}\right) c_{\min} \leq c_a$, each copy of~$D$ is individually feasible with respect to the original capacities. However, we cannot directly interpret each copy as a routing round, because the Modified DGG Algorithm may split the demand of a sink across multiple copies of~$D$ which is incompatible with the notion of routing in rounds. We therefore distinguish two categories of sinks.

A sink is \emph{critically split} if its demand is divided between forward path(s) contained in a single copy of~$D$ and a \emph{singular digraph} that may span multiple copies; otherwise, it is \emph{non-critically split} and its entire demand is routed within one copy, which can be interpreted directly as a round where such a sink is routed.

\begin{observation}\label{obs:one-css-per-copy}
In any fixed copy of~$D$, each arc~$a'$ can be used in the singular digraph of at most one critically split sink. Once used,~$a'$ is deleted and cannot be included in further singular digraphs.
\end{observation}

This observation allows us to group the critically split sinks in a controlled manner. Each such sink~$t$ is associated with a \emph{label}
$\ell\in\{1,\dots,n+1\}$ identifying the copy of~$D$ containing its forward path(s), and a \emph{demand-share vector}~$\theta^t=(\theta^t_1,\dots,\theta^t_{n+1})$, where~$\theta^t_\alpha:=\nicefrac{d_t^\alpha}{d_t}$ denotes the fraction of demand routed through copy~$\alpha$ of~$D$.

The demand-share vectors~$\theta^t$ lie in the~$n$-dimensional simplex. To group sinks of the same label with similar demand distributions, we discretize this simplex using a uniform grid of granularity~$\varepsilon \leq \nicefrac{1}{(n^2-1)(n+1)}$. For each coordinate, the range~$[0,1)$ is sliced into intervals of size~$\varepsilon$. For~$M \coloneqq (n^2-1)(n+1)$, these intervals are of the form~$[k\varepsilon,(k+1)\varepsilon)$ for~$k\in\{0,\dots,M-1\}$. 

A \emph{group} is defined by a label and a choice of one interval per coordinate,
\[
\bigl(\,\ell,\; ([a_1,u_1),\ldots,[a_{n+1},u_{n+1}) \,\bigr)^T,
\qquad
u_\alpha=a_\alpha+\varepsilon \quad \text{for }\alpha=1,\ldots,n+1.
\]
Such a group can contain a demand-share vector if and only if
\[
\sum_{\alpha=1}^{n+1}a_\alpha \leq 1 \quad \text{and} \quad 
1 \leq \sum_{\alpha=1}^{n+1}u_\alpha= \sum_{\alpha=1}^{n+1}(a_\alpha + \varepsilon)= \sum_{\alpha=1}^{n+1}a_\alpha + (n+1)\cdot \varepsilon,
\]
thus, 
$
1-(n+1)\varepsilon \leq \sum_{\alpha=1}^{n+1} a_\alpha \leq 1,
$
or equivalently, by multiplying with~$M$ and using~$a_\alpha=k_\alpha\varepsilon$,  
\[
M-(n+1) \leq \sum_{\alpha=1}^{n+1} k_\alpha \leq M.
\]

Hence, for a fixed label, admissible groups correspond to non-negative integer~$(n+1)$ tuples~$(k_\alpha)$ whose total sum lies in~$\{M - (n+1), \ldots, M  \}$. By the stars-and-bars argument, the total number~$N_{\tot}$ of such groups for a fixed label is
\[
N_{\tot} = \sum_{j=0}^{n+1} \binom{M-j + n}{n} = \sum_{j=0}^{n} \binom{(n^2-1)\cdot(n+1)-j + n}{n} < \infty.
\]

We now show that superimposing all paths of critically split sinks within a fixed group yields a feasible routing round. Fix an arc~$a\in A$. In the worst case, each copy~$\alpha$ of~$a$ carries the maximum group demand~$u_\alpha$ allowed by its interval. By
Observation~\ref{obs:one-css-per-copy}, each copy contributes to at most one singular digraph, so the total flow on~$a$ is strictly less than
\[
\frac{c_a}{(n+1)} + d_{\max}\sum_{\alpha=1}^{n+1} u_\alpha \ 
\]
with the first term coming from the copy that contains the forward path(s) of the sinks of the group. Since
\[
\sum_{\alpha}u_\alpha = \sum_{\alpha}a_\alpha + (n+1)\cdot\varepsilon \leq 1+(n+1)\varepsilon
\quad\text{and}\quad
\varepsilon\le \frac{1}{(n^2-1)(n+1)},
\]
we obtain
\[
\frac{c_a}{(n+1)}+ \left(1- \frac{1}{n}\right)\cdot c_{\min}\cdot \bigl(1 + (n+1)\cdot\varepsilon\bigr) 
\ \leq \ 
c_a \cdot \left( \frac{1}{(n+1)}+ \frac{(n-1)}{n} + \frac{(n^2-1)\cdot \varepsilon}{n}\right) 
\ \leq \ c_a.
\]
This implies that all capacity constraints are satisfied. We conclude that the critically split sinks require at most~$(n+1)\cdot N_{\tot}$ rounds. Adding the~$n+1$ rounds needed for non-critically split sinks, the total
number of routing rounds is~$(n+1)\cdot(N_{\tot}+1)$.

\subsection{Special cases for the maximum routable demand} 
            
When~$d_{\max} < c_{\min}$, the maximum routable demand admits an approximation factor depending on the gap between~$d_{\max}$ and~$c_{\min}$. The boundary case~$d_{\max} = c_{\min}$ remains open, both for routing in rounds and for maximizing routable demand. We now identify settings where additional constraints on source supplies yield constant-factor approximations for the maximum routable demand.

\paragraph{Assuming~$d_{\max} \leq \min_{s \in S^+}b(s)$.} Under this assumption, the derived SSUF instance with super-source~$s^*$ satisfies the balance condition, since each dummy arc~$(s^*, s)$ has capacity~$b(s) \geq d_{\max}$. By~\cite[Corollary 5.4]{DGG}, at least~$0.226\cdot d_{\tot}$ units of demand can be routed in this SSUF instance without violating arc capacities, and thus also in the original unsplittable transshipment instance without exceeding source supplies. Hence, there exists a feasible unsplittable routing that serves at least~$0.226\cdot d_{\tot}$ total demand, with no demand split across sources.

\paragraph{Assuming~$d_{\max} \leq \nicefrac{1}{2} \cdot b^+_{\mathrm{avg}}$.} When~$d_{\max} \leq \nicefrac{c_{\min}}{3}$ (i.e., for \emph{small} demands), there exists a feasible six-round unsplittable routing; see Appendix~\ref{app:finite-rounds_routing}. Hence, in this case, at least~$\nicefrac{1}{6}\cdot d_{\tot}$ can be routed unsplittably in a single round. We therefore restrict our attention to the complementary case of \emph{large} demands in~$[\nicefrac{d_{\max}}{3}, d_{\max}]$, which in particular contains~$[\nicefrac{c_{\min}}{3}, d_{\max}]$.

Assume in the following that demands lie in the range~$[\nicefrac{d_{\max}}{3},\, d_{\max}]$ and that we are additionally given
$
d_{\max} \leq \nicefrac{1}{2}\cdot b^+_{\mathrm{avg}}
= \nicefrac{d_{\tot}}{2|S^+|}.
$
Applying the original DGG algorithm to the derived SSUF instance with a super-source routes all demands unsplittably in four rounds~\cite[Lemma~4.2]{DGG}. When applied to unsplittable transshipment instances, individual sources may be overutilized, but by at most one sink~\cite[Theorem~3.7]{DGG}, so the overutilization at any source is less than~$d_{\max}$. Thus, the total overutilization is less than~$|S^+|\cdot d_{\max}$, which by assumption is at most~$\nicefrac{1}{2}\cdot d_{\tot}$. Therefore, at least~$\nicefrac{1}{2}\cdot d_{\tot}$ units of demand are routed without violating any source supply across the four rounds, and hence at least one round routes a feasible unsplittable flow of value at least
$\nicefrac{1}{4}\cdot \nicefrac{1}{2}\cdot d_{\tot}
= \nicefrac{1}{8}\cdot d_{\tot}$.
Thus, combining the cases of \textit{small} and \textit{large} demands, the maximum routable demand is lower bounded by~$\max \left\{ \nicefrac{d}{6}, \nicefrac{1-d}{8} \right\} \cdot d_{\tot}$ for~$d \in [0,1]$, which attains its minimum at~$d = \nicefrac{3}{7}$ and is therefore always at least~$0.07143\cdot d_{\tot}$.

%----------
% Section 5
%----------

\section{Conclusion and open problems}
\label{sec:conclusion}

We have introduced unsplittable transshipments as a natural extension of single source unsplittable flows. Our results and techniques indicate that unsplittable transshipment problems pose new and interesting algorithmic challenges and are inherently more difficult than their single source counterparts. We conclude by highlighting several open questions and problems that may stimulate future research.

The Modified DGG Algorithm can be adapted to guarantee lower bounds on arc flows in unsplittable transshipments; see Theorem~\ref{thm:lower_bounds}. Whether arc-wise upper and lower bounds can be achieved simultaneously remains an intriguing open question, even in the special SSUF case. Extending a conjecture for SSUF by Morell and Skutella~\cite{MoSk_2022}, we conjecture that such transshipments exist and can be computed efficiently.

\begin{conjecture}
    Given a~$b$-transshipment~$x$, one can efficiently compute an unsplittable~$b$-transshipment~$y$ such that
    \[x_a - d_{\max} \ \leq \ y_a \ \leq \ x_a + d_{\max} \quad \text{for all } a\in A.\]
\end{conjecture}

Motivated by a famous conjecture of Goemans, many authors have studied cost-based variants of SSUF; see, e.g.,~\cite{Kolliopoulos_3_2_SSUF,Skutella_Apprx_SSUF_3_1,traub2024single,swamy2026unsplittable}. Generalizing some of these results to unsplittable transshipments poses an interesting challenge.

In Section~\ref{sec:beyond_congestion}, we study the problem of minimizing the number of rounds while imposing the balance condition~$d_{\max}\leq c_{\min}$. Our techniques exploit the additional restriction to~$d_{\max}< c_{\min}$, allowing one to sufficiently separate the demands from~$c_{\min}$. It remains open whether explicit bounds can be obtained even if~$d_{\max} = c_{\min}$, both for routing in rounds and for maximizing the routable demand. Other variants, such as constant-factor approximations for minimizing the number of rounds, may also provide further insight.

Another related model is the \emph{$k$-splittable flow} problem, introduced by Baier, Köhler, and Skutella \cite{baier2005k}, in which given supplies and demands must be routed using at most~$k$ paths per source--sink pair. They studied the maximum $k$-splittable $s$--$t$ flow problem, establishing tight approximation guarantees for small values of~$k$ and polynomial-time algorithms for maximum uniform $k$-splittable $s$--$t$ flows. Subsequent work investigated congestion- and cost-based variants in the single source setting, including approximation results for~$k=2$ by Kolliopoulos~\cite{kolliopoulos2005minimum} and for general~$k$ by Salazar and Skutella~\cite{salazar2009single}.

The model extends naturally to \emph{$k$-splittable $b$-transshipments}, where each source--sink pair may use at most~$k$ paths. By pre-splitting sink demands and applying the Modified DGG Algorithm, one obtains congestion bounded by $\nicefrac{d_{\max}}{k}$ for the $k$-splittable transshipment problem. Other optimization variants that arise in the context of $k$-splittable flows remain open for the $k$-splittable transshipment setting, suggesting several promising directions for future research.
   
%\newpage

%----------
% Bibligraphy
%----------

\printbibliography 
% in footnotesize (see format.tex to change)

%----------
% Appendix
%----------

\begin{appendices}

\section{Hardness result}
\label{app:hardness}

\begin{lemma}
    Given an unsplittable transshipment instance with integral supplies and demand values and unit arc capacities, it is \texttt{NP}-complete to decide whether there exists a feasible unsplittable transshipment.
\end{lemma}

\begin{proof}
The problem clearly belongs to~\texttt{NP}. To establish \texttt{NP}-hardness, we reduce the \texttt{NP}-complete \emph{arc-disjoint paths problem}~\cite{knuth1974} which is defined as follows: Given a directed graph~$D=(V,A)$ and a set~$\{(s^i,t^i)\}_{i=1}^k \subseteq V \times V$ of vertex pairs, the task is to determine whether there exist pairwise arc-disjoint paths~$P_1,\dots,P_k$ such that~$P_i$ is an~$s^i$--$t^i$ path for each~$i = 1,\dots,k$. 

A given arc-disjoint paths instance can be turned into an unsplittable transshipment instance with unit arc capacities as follows: add direct arcs~$(s^i,t^j)$ for all~$i\neq j$ and assign supply~$b(s^i)=k$ and demand~$b(t^i)=-k$, for~$i = 1,\dots,k$. It is then easy to check that there exists a feasible unsplittable transshipment if and only if there exist arc-disjoint~$s^i$--$t^i$ paths for~$i=1,\dots,k$.
\end{proof}

\begin{corollary}
Given an unsplittable transshipment instance, deciding whether there exists a feasible unsplittable transshipment in which all paths leaving the same source and all paths entering the same sink carry identical flow values is \texttt{NP}-complete.
\end{corollary}

\section{Non-integrality result}
\label{app:non-integrality}

Figure~\ref{fig:int_UT_instance_frac} presents an unsplittable transshipment instance with integral arc capacities and integral supplies and demands, following the construction in~\cite[Figure 4]{baier2005k}. The instance admits a feasible non-integral unsplittable transshipment, as shown in the figure. However, no feasible integral unsplittable transshipment exists. To see this, notice that the total demand is~$20$, and each path has capacity at most~$9$, so at least~$4$ paths are required. At vertex~$v$ (resp.~$w$), two paths must route flow upward and two downward. 

Consider the possibilities for integral flows along these paths: If any path carries flow strictly less than~$1$, then some other path must carry more than~$9$, violating capacity constraints. If all paths carry flow strictly greater than~$>2$, arcs of capacity~$2$ cannot be used, making the transshipment infeasible. If a path carries exactly~$1$ (resp.~$2$), another path must carry~$9$ (resp.~$8$), forcing a third to carry~$3=12-9$ (resp~$4=12-8$) along the arc of capacity~$12$. The remaining path would then carry~$7$ (resp.~$6$), again exceeding capacity at the arc of capacity~$15$ (resp.~$5$). 

Thus, at least one path must carry flow strictly between~$1$ and~$2$, implying that no integral unsplittable transshipment can satisfy all balance values and capacities.

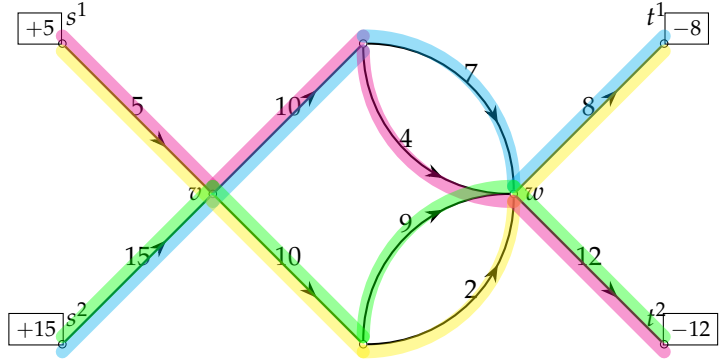
\begin{figure}[t]
\centering
\tikzset{midarrow/.style={postaction={decorate},decoration={markings,mark=at position 0.7 with {\arrow{Stealth}}}}}

\tikzset{
    dot/.style={circle,draw,inner sep=1pt},
    lab/.style={draw=none,font=\small,inner sep=1pt},
    midarrow/.style={postaction={decorate,decoration={markings,mark=at position 0.7 with {\arrow{Stealth}}}}},
    arc/.style={thick,midarrow},
     highlight/.style={draw, line width=5pt, opacity=0.4, line cap=round}
}

\begin{tikzpicture}

% Sources
\node[dot] (s1) at (0,4) {}; 
\node[draw, scale=0.75] at ($(s1)+(-0.3,0.2)$) {$+5$};
\node[lab, above] at ($(s1)+(0.2,0.2)$) {$s^1$};

\node[dot] (s2) at (0,0) {}; 
\node[draw, scale=0.75] at ($(s2)+(-0.35,0.2)$) {$+15$};
\node[lab, above]at  ($(s2)+(0.2,0.2)$) {$s^2$};

% Middle Nodes
\node[dot] (v1) at (2,2) {}; 
\node[lab, left]at ($(v1)+(-0.1,0)$) {$v$};
\node[dot] (v2) at (6,2) {}; 
\node[lab, right]at ($(v2)+(0.1,0)$) {$w$};

\node[dot] (u) at (4,4){};
\node[dot] (l) at (4,0){};

% Sinks
\node[dot] (t1) at (8,4) {};  
\node[draw, scale=0.75] at ($(t1)+(0.3,0.2)$) {$-8$};
\node[lab, above]at ($(t1)+(-0.1,0.2)$) {$t^1$};

\node[dot] (t2) at (8,0) {};
\node[draw, scale=0.75] at ($(t2)+(0.35,0.2)$) {$-12$};
\node[lab, above]at ($(t2)+(-0.1,0.2)$) {$t^2$};

% Arcs Network

\draw[arc] (s1) -- node[lab,above] {$5$}  (v1);
\draw[arc] (s2) -- node[lab,above] {$15$} (v1);

\draw[arc] (v1) -- node[lab,above] {$10$} (u);
\draw[arc] (v1) -- node[lab,above] {$10$} (l);

\draw[arc] (u) to[bend left=45] node[lab,above] {$7$} (v2);
\draw[arc] (u) to[bend right=45] node[lab,above] {$4$} (v2);
\draw[arc] (l) to[bend left=45] node[lab,above] {$9$} (v2);
\draw[arc] (l) to[bend right=45] node[lab,above] {$2$} (v2);

\draw[arc] (v2)-- node[lab,above] {$8$} (t1);
\draw[arc] (v2)-- node[lab,above] {$12$} (t2);

% Colors

% s1 -> v1
\draw[highlight,color=magenta] (0,4.1) -- (2,2.1);
\draw[highlight,color=yellow] (0,3.9) -- (2,1.9);

% s2 -> v1
\draw[highlight,color=cyan] ($ (s2)-(0,0.1) $) -- ($ (v1)-(0,0.1) $);
\draw[highlight,color=green] ($ (s2)+(0,0.1) $) -- ($ (v1)+(0,0.1) $);

% v1 -> u
\draw[highlight,color=magenta] ($ (v1)+(0,0.1) $) -- ($ (u)+(0,0.1) $);
\draw[highlight,color=cyan] ($ (v1)-(0,0.1) $) -- ($ (u)-(0,0.1) $);

% v1 -> l
\draw[highlight,color=yellow] ($ (v1)-(0,0.1) $) -- ($ (l)-(0,0.1) $);
\draw[highlight,color=green] ($ (v1)+(0,0.1) $) -- ($ (l)+(0,0.1) $);

% u -> v2 
\draw[highlight,color=cyan] ($ (u)+(0,0.1) $) to[bend left=45] ($ (v2)+(0,0.1) $);
\draw[highlight,color=magenta] ($ (u)-(0,0.1) $) to[bend right=45] ($ (v2)-(0,0.1) $);

% l -> v2 
\draw[highlight,color=green] ($ (l)+(0,0.1) $) to[bend left=45] ($ (v2)+(0,0.1) $);
\draw[highlight,color=yellow] ($ (l)-(0,0.1) $) to[bend right=45] ($ (v2)-(0,0.1) $);

% v2 -> t1
\draw[highlight,color=cyan] ($ (v2)+(0,0.1) $) -- ($ (t1)+(0,0.1) $);
\draw[highlight,color=yellow] ($ (v2)-(0,0.1) $) -- ($ (t1)-(0,0.1) $);

% v2 -> t2
\draw[highlight,color=green] ($ (v2)+(0,0.1) $) -- ($ (t2)+(0,0.1) $);
\draw[highlight,color=magenta] ($ (v2)-(0,0.1) $) -- ($ (t2)-(0,0.1) $);

\end{tikzpicture}
\caption{A non-integral unsplittable transshipment, formed by the~$s^1$-$t^1$-path with flow value~$\nicefrac{3}{2}$ highlighted in \emph{yellow}, the~$s^1$-$t^2$-path with flow value~$\nicefrac{7}{2}$ highlighted in \emph{magenta}, the~$s^2$-$t^1$-path with flow value~$\nicefrac{13}{2}$ highlighted in \emph{blue}, and the~$s^2$-$t^2$-path with flow value~$\nicefrac{17}{2}$ highlighted in \emph{green}.}
\label{fig:int_UT_instance_frac}
\end{figure}

\section{Tightness of upper bound in highly connected instances}\label{app:sec:tightness}
            
We show that~$d_{\max}$ remains tight even when each sink is connected to an arbitrarily large number ($k+1$) of sources, by constructing the following family of instances parameterized by any integer~$q \geq 3$ and~$q>(k+1)$; see Figure~\ref{fig:d_max-LB_arb_alt} for an illustration.

Our instance consists of~$q$ sinks~$t^1,\ldots,t^q$, each with demand~$d_j=1$. For each~$j\in\{1,\ldots,q\}$, there is a dedicated source~$s^j$ connected to~$t^j$ by~$q-k$ parallel arcs. In addition, there are~$k$ sources~$s^{q+1},\ldots, s^{q+k}$, each connected to every sink~$t^j$ by a path of length~$2$, which consists of an arc~$a_i$  shared across all sinks and one arc entering~$t^j$. The balance values of the sources and the sinks can be found in Figure~\ref{fig:d_max-LB_arb_alt}. We now define a feasible fractional~$b$-transshipment as follows: every arc entering a sink carries flow~$\nicefrac{1}{q}$. Thus, each sink receives a total flow of~$1$, and each shared arc~$a_i$ with~$i\in\{1,\ldots,k\}$ carries a total flow of~$q\cdot\nicefrac{1}{q}=1$.

Consider any unsplittable~$b$-transshipments of our instance. Flow from source~$s^j$ can only be routed to sink~$t^j$, and unsplittability requires to route its entire flow value along one arc. This yields a violation of~$1-\nicefrac{(k+1)}{q}$, which converges to~$d_{\max}=1$ as~$q\to\infty$. Thus, even when each sink is connected to arbitrarily many sources, no unsplittable transshipment can avoid a violation of~$d_{\max}$, establishing the tightness of the bound.

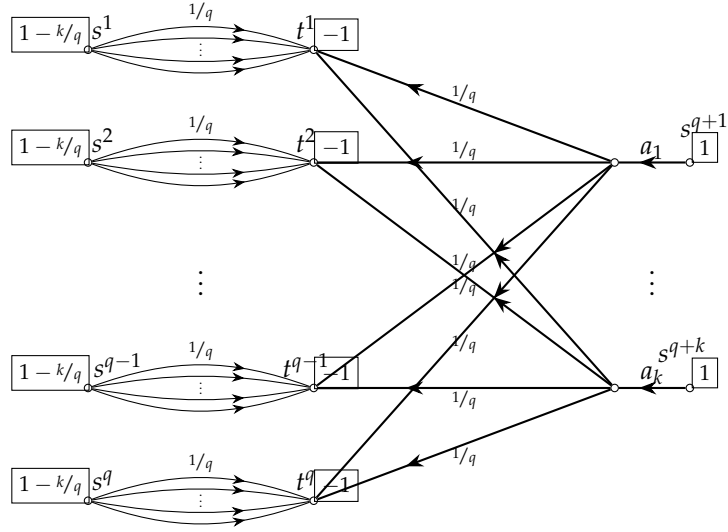
\begin{figure}[t]     
    \centering
    \tikzset{
    dot/.style={circle,draw,inner sep=1pt},
    lab/.style={draw=none,font=\small,inner sep=1pt},
    arc/.style={thick,midarrow},
     arcx/.style={thick,midarrowx},
    midarrow/.style={postaction={decorate,decoration={markings,mark=at position 0.7 with {\arrow{Stealth}}}}},
    midarrowx/.style={postaction={decorate,decoration={markings,mark=at position 0.4 with {\arrow{Stealth}}}}}
}

\begin{tikzpicture}

% Source-sink pairs

\foreach \j/\y in {1/6,2/4.5,3/1.5,4/0} {
    \node[dot]  (s\j) at (2,\y) {};
    \node[draw, scale=0.75] at ($(s\j)+(-0.5,0.2)$) {$1-\nicefrac{k}{q}$};
    \node[dot] (t\j) at (5,\y) {};   
    \node[draw, scale=0.75] at ($(t\j)+(0.3,0.2)$) {$-1$};
    % q-k parallel arcs of capacity 1/q
    \foreach \i in {-20,-10,10,20} {
        \draw[midarrow,bend left=\i] (s\j) to (t\j);
    } 
    \node[lab, scale=0.5] at (3.5,\y+0.05) {$\vdots$};  
    \node[lab, scale=0.75] at (3.5,\y+0.5) {$\nicefrac{1}{q}$};
}

\node[lab] at (3.5,3) {$\vdots$};

\node[lab,above]  at  ($(s1)+(0.2,0.1)$) {$s^{1}$};
\node[lab,above]  at  ($(s2)+(0.2,0.1)$) {$s^{2}$};
\node[lab,above]  at  ($(s3)+(0.4,0.1)$) {$s^{q-1}$};
\node[lab,above]  at  ($(s4)+(0.2,0.1)$) {$s^{q}$};

\node[lab,above]  at ($(t1)+(-0.1,0.1)$) {$t^{1}$};
\node[lab,above]  at ($(t2)+(-0.1,0.1)$) {$t^{2}$};
\node[lab,above]  at ($(t3)+(-0.1,0.1)$) {$t^{q-1}$};
\node[lab,above]  at ($(t4)+(-0.1,0.1)$) {$t^{q}$};

%\node[rectangle,right] at ($(t\j.east)+(0.5,0)$) {$d_{t^{\j}}=1$};
    % single shared capacity label for the bundle
%    \node[lab, font=\tiny, inner sep=1pt] at ($(s\j)!0.5!(t\j)+(0,0.3)$) {$1/q$};

% Auxiliary sources

\node[dot] (sa1) at (10,4.5) {};
\node[lab,above]  at  ($(sa1)+(0.2,0.3)$) {$s^{q+1}$};
 \node[draw, scale=0.75] at ($(sa1)+(0.2,0.2)$) {$1$};

\node[dot] (sak) at (10,1.5) {};
\node[lab,above]  at  ($(sak)+(-0.1,0.3)$) {$s^{q+k}$};
 \node[draw, scale=0.75] at ($(sak)+(0.2,0.2)$) {$1$};

\node[lab] at (9.5,3) {$\vdots$};

% Shared arcs a_i

\node[dot] (v1) at (9,4.5) {};
\node[dot] (vk) at (9,1.5) {};

\draw[arc] (sa1) -- node[lab,above] {$a_1$}  (v1);
\draw[arc] (sak) -- node[lab,above] {$a_{k}$}  (vk);

% Length-2 paths to all sinks

\foreach \j in {1,2} {
    \draw[arc] (v1) -- node[lab,above,scale=0.75] {$\nicefrac{1}{q}$}  (t\j);
    \draw[arcx] (vk) -- node[lab,above, scale=0.75] {$\nicefrac{1}{q}$}  (t\j);
}
\foreach \j in {3,4} {
    \draw[arc] (vk) -- node[lab,below,scale=0.75] {$\nicefrac{1}{q}$}  (t\j);
    \draw[arcx] (v1) -- node[lab,below, scale=0.75] {$\nicefrac{1}{q}$}  (t\j);
}

\end{tikzpicture}
    \caption{An instance to show~$d_{\max}$ is a tight lower bound for the required capacity violation.}
    \label{fig:d_max-LB_arb_alt}   
\end{figure}

\section{Proof of Theorem~\ref{thm:tree}}
\label{app:sec:number_of_paths}
 
To prove Theorem~\ref{thm:tree}, we need to specify Step~\ref{step:singular} of Algorithm~\ref{alg:modifiedDGG} in more detail. Consider the situation where the algorithm finds a singular digraph rooted at the junction vertex~$w$, and~$w$ has been reached on a forward path along its incoming arc~$a'$. Moreover, let~$t$ be a (sub-)sink located at~$w$ with~$d_t > y_{a'}$.

By Lemma~\ref{lem:singular-digraph}, the singular digraph with the super-source~$s^*$ removed forms an in-tree rooted at~$w$. The flow~$y$ on this in-tree satisfies the demands of all (sub-)sinks located at its vertices, with the sole exception of the sink~$t$, for which only a demand of~$d_t - y_{a'}$ is satisfied within the in-tree, while the remaining demand~$y_{a'}$ is supplied via arc~$a'$.

We proceed iteratively as follows. Starting from a leaf~$s\in S^+$ of the in-tree, we follow its outgoing arc~$a$ and continue along the unique path toward the root~$w$ until a (sub-)sink~$t'$ is encountered. If~$t'$ is located at the root~$w$, we select~$t' = t$ only if no other (sub-)sink is located at~$w$. If~$d_{t'}\leq y_a$, the entire demand~$d_{t'}$ can be routed along this path. In this case, we move~$t'$ backward along the path to~$s$ and further to the super-sink~$s^*$, reducing the flow on all arcs of the path by~$d_{t'}$. Otherwise, if~$d_{t'} > y_a$, we split~$t'$ into two sub-sinks~$t'_1$ and~$t'_2$ with demands~$d_{t'_1} \coloneqq d_{t'} - y_a$ and~$d_{t'_2} \coloneqq y_a$, respectively. The sub-sink~$t'_2$ is then moved backward to~$s$ and further to the super-sink~$s^*$, and the flow on all arcs of the path is reduced by~$y_a$. As a consequence, the arc~$a$ is deleted and~$s$ has no remaining outgoing arcs. We thus delete~$s$ from the in-tree. 

This process is repeated until all flow on the singular digraph has been eliminated and only the sub-sink~$t_1$ remains at the root~$w$, after which it is moved backward along the arc~$a'$.

With this additional specification of the Modified DGG Algorithm, we are now ready to prove Theorem~\ref{thm:tree}.

\begin{proof}[Proof of Theorem~\ref{thm:tree}.]
We prove the result by analyzing how the bipartite graph~$G$ evolves during the execution of the Modified DGG Algorithm. Whenever a sink~$t\in S^-$, or one of its sub-sinks, reaches a source~$s \in S^+$ along its backward path to the super-source~$s^*$, we add the edge~$\{s,t\}$ to~$G$.

A source~$s \in S^+$ is called \emph{active} as long as it has at least one outgoing arc. Similarly, a sink~$t \in S^-$ is called \emph{active} as long as~$t$, or one of its sub-sinks, has not yet reached the super-source~$s^*$. Note that at any time there is at most one (sub-)sink of~$t$ that has not reached the super-source. Indeed, whenever a (sub-)sink is split into two sub-sinks, one of them is immediately moved along a backward path within a singular digraph, whose arcs are then deleted.

We prove that the following invariants are maintained throughout the execution of the Modified DGG Algorithm:
\begin{enumerate}
\item[(*)]~$G$ is acyclic and every connected component of~$G$ contains at most one active source or sink.
\end{enumerate}
After the preliminary phase of the algorithm, some sinks may have reached the super-source~$s^*$ via a source in~$S^+$. These sinks are inactive and have exactly one incident edge in~$G$, connecting them to their corresponding source in~$S^+$. All remaining sinks are still active and isolated in~$G$. In particular, the invariant~(*) is satisfied at this point.

Throughout the execution of the algorithm, the bipartite graph changes only when a sink~$t$, or one of its (sub-)sinks, reaches a source~$s \in S^+$ and is then immediately moved further backward to the super-source. Prior to this step, both~$t$ and~$s$ are active and therefore belong to distinct connected components of~$G$, which are merged by the addition of the edge~$\{s,t\}$. In particular,~$G$ is still acyclic. It only remains to show that~$s$ or~$t$ are inactive after this step such that the new connected component of~$G$ formed by edge~$\{s,t\}$ contains at most one active source or sink.

If~$t$ has become inactive, we are done. Otherwise,~$t$ was split into two sub-sinks~$t_1$ and~$t_2$ in the previous step, and only~$t_2$ was moved to the super-sink~$s^*$ via a source~$s \in S^+$. In this case, the demand of~$t_2$ was set equal to the flow value~$y_a$ of the unique outgoing arc~$a$ of~$s$, so that~$a$ was deleted and~$s$ consequently became inactive.
\end{proof}
        
\section{Routing in a constant number of rounds}
\label{app:finite-rounds_routing}

\paragraph{Routing in six rounds when~$d_{\max} \leq \nicefrac{c_{\min}}{3}$.}
As in Section~\ref{subsec:number_of_rounds}, construct an auxiliary graph~$D'$ with~$q=2$ copies of~$D$, each carrying half of the original flow. Applying the Modified DGG Algorithm yields a feasible routing in each copy, since~$\nicefrac{c_a}{2}+d_{\max}<c_a$. These two copies directly serve as two routing rounds for all non-critically split sinks.
        
Each critically split sink is labeled by the copy containing its forward path(s). For sinks with the same label, we further distinguish two cases according to how much of its demand is routed along the singular digraph: \emph{Type~$1$}, at least half, i.e.~$\ge \nicefrac{d_t}{2}$; and \emph{Type~$2$}, less than half, i.e.~$< \nicefrac{d_t}{2}$. This produces four groups in total (two labels times two types). We route each group in a separate additional round, yielding four extra rounds and six rounds overall.
        
\emph{Feasibility.} Consider a group labeled by Copy~$1$ (the case of Copy~$2$ is symmetric). In its routing round, the singular digraph (lying in Copy~$2$) contributes at most~$d_{\max}$ flow per arc for Type~$1$ and at most~$\nicefrac{d_{\max}}{2}$ for Type~$2$, while the forward paths (lying in Copy~$1$) contribute at most~$\nicefrac{c_a}{2}+\nicefrac{d_{\max}}{2}$ for Type~$1$ and~$\nicefrac{c_a}{2}+d_{\max}$ for Type~$2$. In either case, the total flow on any arc~$a$ is bounded by
\[
\nicefrac{c_a}{2}+\nicefrac{3d_{\max}}{2} \;\le\; c_a,
\]
since~$d_{\max}\le \nicefrac{c_{\min}}{3}\le \nicefrac{c_a}{3}$. Hence, each additional round is feasible. Therefore, if $d_{\max}\le \nicefrac{c_{\min}}{3}$, then all demands can be routed unsplittably in six rounds.

\paragraph{Routing in four rounds when~$d_{\max} \leq \nicefrac{c_{\min}}{4}$.}
As above, we build the auxiliary graph~$D'$ from two copies of~$D$ and run the Modified DGG algorithm. Sinks that are not critically split are already routed feasibly in one of the two copies (rounds), since~$\nicefrac{c_a}{2} + d_{\max} \le \nicefrac{c_a}{2} + \nicefrac{c_{\min}}{4} < c_a$.

Each critically split sink is labeled by the copy of~$D$ containing its forward path(s). Critically split sinks with the same label are grouped, and their paths are superimposed to form two additional rounds. No finer grouping is required for feasibility as~$\nicefrac{c_a}{2} + d_{\max} + d_{\max} \leq \nicefrac{c_a}{2} + \nicefrac{c_{\min}}{2} \leq c_a$.

\end{appendices}     
\end{document}